\documentclass[reqno,11pt]{amsart}
\usepackage{cases}
\usepackage{mathrsfs}
\usepackage[T1]{fontenc}
\usepackage{mathrsfs}
\usepackage{amsmath,latexsym,amssymb,amsfonts,amsbsy, amsthm}
\usepackage[usenames]{color}
\usepackage{xspace,colortbl}
\usepackage{epsfig}
\usepackage{amsmath,amsfonts,amsthm,amssymb,amscd}
\input amssym.def
\input amssym.tex
\allowdisplaybreaks[4]
 \usepackage{subfigure}
\newtheorem{theorem}{Theorem}[section]

\newtheorem{lemma}[theorem]{Lemma}

\newcommand{\ben}{\begin{eqnarray}}
\newcommand{\een}{\end{eqnarray}}
\newcommand{\beno}{\begin{eqnarray*}}
\newcommand{\eeno}{\end{eqnarray*}}

\newdimen\eqjot \eqjot = 1\jot
\def\openupeq{\openup \the\eqjot}
\def\addtab#1={#1\;&=}
\def\addtabe#1=#2={#1=#2\;&=}

\begin{document}

\title []
{Pfaffian solution for dark-dark soliton to the coupled complex modified Korteweg-de Vries equation}

\author{Chenxi Li}
\address{Chenxi Li\newline
School of Mathematics and Statistics, Xi'an Jiaotong University, Xi'an, Shaanxi 710049, P.R. China}
\email{lichenxi\_math@stu.xjtu.edu.cn}
\author{Xiaochuan Liu}
\address{Xiaochuan Liu\newline
School of Mathematics and Statistics, Xi'an Jiaotong University, Xi'an, Shaanxi 710049, P.R. China}
\email{liuxiaochuan@mail.xjtu.edu.cn}
\author{Bao-Feng Feng}
\address{Corresponding author: Bao-Feng Feng\newline
 School of Mathematics and Statistical Sciences, The University of Texas Rio Grande Valley, Edinburg, TX 78539, USA}
\email{baofeng.feng@utrgv.edu}
\date{\today}

\begin{abstract}
In this paper, we study coupled complex modified Korteweg-de Vries  (ccmKdV) equation by combining the Hirota's method and the Kadomtsev-Petviashvili (KP) reduction method. First, we show that the bilinear form of the ccmKdV equation under nonzero boundary condition is linked to the discrete BKP hierarchy through Miwa transformation. Based on this finding, we construct the dark-dark soliton solution in the pfaffian form. The dynamical behaviors for one- and two-soliton are analyzed and illustrated.
\end{abstract}

\maketitle \numberwithin{equation}{section}




\section{Introduction}

Higher-order nonlinear Schr\"odinger (HONLS) equation
\begin{equation}\label{gNLS}
	\mathrm{i} q_t +\alpha_1 q_{xx} + \alpha_2 |q|^2 q + \mathrm{i} ( \beta_1 q_{xxx} + \beta_2 |q|^2 q_x +\beta_3 q (|q|^2)_x ) = 0\,,
\end{equation}
was first developed by Kodama and Hasegawa in the nonlinear optics \cite{kodama1987nonlinear}. The parameter $\alpha _{2}>0$ denotes the Kerr effect-induced self-phase modulation(SPM) , while $\alpha _{1}$ governs the group velocity dispersion (GVD): $\alpha _{1}>0$ is focusing case \cite{shabat1972exact}, whereas for $\alpha _{1}<0$ is defocusing case \cite{zakharov1973interaction}. Here, $\beta _{1}$, $\beta _{2}$, and $\beta _{3}$ determine three fundamental higher-order effects: third-order dispersion, self-steepening, and stimulated Raman scattering, respectively \cite{trippenbach1998effects}.
In some special cases, Eq.\,(\ref{gNLS}) becomes integrable when specific constraints on parameters $\beta _{i}$ are imposed. For instance, several well-known integrable models have emerged: (\romannumeral1) the Kaup-Newell equation \cite{Kaup_Newell} ($\beta_{1}=0,\beta _{2}:\beta _{3}=1:1$), (\romannumeral2) the Chen-Lee-Liu equation \cite{Chen_1979} ($\beta_{1}=\beta _{3}=0$), (\romannumeral3)
the Hirota equation \cite{Hirota1973Eqn} ($\beta_{1}: \beta _{2}:\beta_{3}=1:6:0$
), (\romannumeral4) the Sasa-Satsuma (SS) equation \cite{sasa1991new} ($\beta_{1}:\beta _{2}:\beta_{3}=1:6:3$).

Due to the polarization of propagation pulses, the coupled models of the HONLS are important in practical applications \cite{agrawal2000nonlinear,menyuk1987nonlinear,wai1991effects}. In other words, a coupled NLS equation with extra dispersion and nonlinear terms is more practical in nonlinear optics. As discussed in \cite{gilson2003sasa}, there are several two-component generalizations of the HONLS equation which are integrable . These integrable models include
\begin{enumerate}
	\item The coupled Hirota equation \cite{tasgal1992soliton}
	\begin{eqnarray}
		&&    u_{1,t}=u_{1,x x x}-3c(|u_1|^{2}+|u_2|^{2}) u_{1,x}-3c u_1 (u^*_1 u_{1,x} +u^*_2 u_{2,x}), \\
		&& u_{2,t}=u_{2,x x x}-3c(|u_1|^{2}+|u_2|^{2}) u_{2,x} -3c u_2 (u^*_1 u_{1,x} +u^*_2 u_{2,x}).
	\end{eqnarray}
	\item 
	The coupled Sasa-Satsuma (CSS) equation \cite{porsezian1994coupled}
	\begin{eqnarray}
		&&u_{1,t}=u_{1,x x x}-3c(|u_1|^{2}+|u_2|^{2}) u_{1,x}-3c u_1 (|u_1|^{2}+|u_2|^{2})_x, \\
		&&u_{2,t}=u_{2,x x x}-3c(|u_1|^{2}+|u_2|^{2}) u_{2,x}-3c u_2 (|u_1|^{2}+|u_2|^{2})_x. 
	\end{eqnarray}
	\item The coupled complex modified Korteweg-de Vries (ccmKdV) equation \cite{sakovich2000symmetrically}
	\begin{eqnarray}
		&& u_{1,t}=u_{1,x x x}-3c(|u_1|^{2}+|u_2|^{2}) u_{1,x},    \label{ccmkdv1}\\
		&& u_{2,t}=u_{2,x x x}-3c(|u_1|^{2}+|u_2|^{2}) u_{2,x}  \label{ccmkdv2}\,.
	\end{eqnarray}
\end{enumerate}

There are much more studies for the coupled Hirota equation than for the CSS and the ccmKdV equations probably due to the reason that it is the simplest equation among three integrable cases (see, for example, Refs. \cite{wang2021analytical,liu2021long,kang2019construction,chen2013rogue,wang2014generalized}). In the original paper by Tasgal et al.\,\cite{tasgal1992soliton}, the bright soliton was obtained by using the inverse scattering transformation.
Bright, dark and bright-dark soliton solutions were constructed by various researchers  \cite{wang2021analytical,park2000higher,bindu2001dark,porsezian1997optical}. Recently, rogue wave solutions of the coupled Hirota equation were developed in \cite{chen2013rogue,chen2014dark,wang2014generalized,wang2014rogue1}. 

The soliton solutions of the CSS equation were  mainly constructed by the Darboux transformation. In \cite{xu2013single}, the authors constructed the single- and double-hump solutions under the zero boundary condition, which were expressed as Wronskian determinants. Bright multi-soliton solutions were derived and the energy transfer mechanism was revealed in \cite{lu2014bright,liu2018vector}. The rational solutions in localized conditions under non-zero boundaries were established in \cite{zhao2014localized}, in which the authors revealed the dark-antidark soliton, W-shape solution, Mexican
hat, and anti-Mexican hat solutions. In \cite{zhang2017binary}, both bright and dark solitons were studied, with single and bright hump solutions found on a vanishing background and an anti-dark soliton solution found on a non-vanishing background. In \cite{liu2018dark}, the dark-bright soliton and semirational rogue wave solutions were derived via the Darboux-dressing method.

On the other hand, the Riemann-Hilbert problem to the CSS equation was formulated and studied in \cite{geng2016riemann,liu2023riemann,wu2022riemann,wu2017inverse}, in which the soliton solutions and their asymptotic properties are reported. Most recently, Zhang et al. derived the bilinear form of the CSS equation and constructed various soliton solutions such as the dark,  breather, and rogue wave solutions by the Kadomtsev-Petviashvili (KP) reduction method \cite{zhang2025dark,zhang2025rogue}. Moreover, the bright-bright, dark-dark, and bright-dark soliton solutions are obtained by applying the KP reduction method \cite{shi2025general}. 

In spite of the fact that the ccmKdV equation (\ref{ccmkdv1})--(\ref{ccmkdv2}) is a straightforward generalization of the complex mKdV equation 
\begin{equation}
	u_{t}=u_{x x x}-3 c|u|^{2} u_{x},
\end{equation}
the associated study related to the ccmKdV equation is much less compared with the coupled Hirota and CSS equations. The pfaffian form of the general bright soliton solution to the multi-component mKdV equation was given by Iwao and Hirota \cite{iwao1997soliton}. 
Tsuchida studied the coupled mKdV equation using the inverse scattering transformation \cite{tsuchida1998coupled}. By using Hirota's bilinear method, one- and two-bright soliton and breather solutions \cite{xu2023breather}, as well as the first-order rogue wave solution, were constructed in \cite{chan2017rogue}.

A $4 \times 4$ Lax pair of the ccmKdV equation (\ref{ccmkdv1})--(\ref{ccmkdv2}) was given by Tsuchida \cite{tsuchida1998coupled} as $\Phi_x=U\Phi$, $\Phi_t=V\Phi$ with
\begin{align*}
	U=\mathrm{i} \lambda \begin{pmatrix}
		-I & O  \\ O & I
	\end{pmatrix}
	+
	\begin{pmatrix}
		O &   Q  \\ R & O
	\end{pmatrix},
\end{align*}
\begin{align*}
	V=&\mathrm{i} \lambda^3 \begin{pmatrix}
		-4I & O  \\ O & 4I
	\end{pmatrix}   
	+
	\lambda^2 \begin{pmatrix}
		O & 4Q \\ 4R & O
	\end{pmatrix}
	+
	\mathrm{i} \lambda \begin{pmatrix}
		-2QR & 2Q_x \\ -2R_x & 2RQ
	\end{pmatrix}\\
	&+
	\begin{pmatrix}
		Q_xR-QR_x & -Q_{xx}+2QRQ  \\ -R_{xx}+2RQR & R_xQ-RQ_x
	\end{pmatrix},
\end{align*}
\begin{align*}
	Q=\sqrt{\dfrac{c}{2}} \begin{pmatrix}
		u_1 & u_2  \\ -u^*_2 & u^*_1
	\end{pmatrix},\ \ 
	R=\sqrt{\dfrac{c}{2}} \begin{pmatrix}
		u^*_1 & -u_2  \\ u^*_2 & u_1
	\end{pmatrix},
\end{align*}
where $\lambda$ is the spectral parameter that is independent of time. $u_1$ and $u_2$ are complex-valued functions. $I$  is the $2 \times 2$ identity matrix. $O$ is the $2 \times 2$ zero matrix. It is noticed that system (\ref{ccmkdv1})--(\ref{ccmkdv2}) also admits a $6\times6$ Lax pair 
in \cite{Adamopoulou2019DrinfeldSokolovCA}.

Due to the pfaffian structure of the soliton solution to the ccmKdV equation, its multi-dark soliton solution under nonzero boundary condition remains an unsolved problem, which is the motivation of the present paper.   
The remainder of the paper is organized as follows: In Section 2, the main result for the dark-dark soliton solution in pfaffian is presented by bilinearizing the ccmKdV equation (\ref{ccmkdv1})--(\ref{ccmkdv2}).   The proof of the main result is given in Section 3. The dynamics analysis for one- and two-soliton solutions is constructed in Section 4.

\section{Bilinearization and dark soliton solutions under nonzero boundary condition}


In this section, we will give the bilinear form of the coupled complex mKdV equation (\ref{ccmkdv1})--(\ref{ccmkdv2})
under a nonzero boundary condition. To this end, we assume the following transformations
\begin{equation}\label{coupled HH}
	u_1=\rho_1 \frac{g_1}{f} e^{\mathrm{i}\left(\alpha_1 x -\omega_1 t\right)}, \quad   u_2= \rho_2 \frac{g_2}{f} e^{\mathrm{i}\left(\alpha_2 x -\omega_2 t\right)},
\end{equation}
where $f$ is a real-valued function, $g_1,g_2$ are complex-valued functions, and $\rho_i$ and $\alpha_i$ are real parameters, $\omega_i=\alpha^3_i + 3 c \alpha_i (\rho_1^2+\rho_2^2), i=1,2.$ 

By substituting (\ref{coupled HH}) into (\ref{ccmkdv1}), we obtain
\begin{align}\nonumber
	&f^2( D^3_{x} - D_{t} +3 \mathrm{i} \alpha_1  D^2_{x} -3 \alpha^2_1 D_{x} -3c(\rho_1^2+\rho_2^2)D_{x})\,g_1 \cdot f\\
	&-3 (D_{x} g_1 \cdot f)(D^2_{x} f \cdot f+c (\rho_1^2 |g_1|^2+\rho_2^2 |g_2|^2)-c(\rho_1^2+\rho_2^2) f^2 ) \\\nonumber
	&-3 \mathrm{i} \alpha_1  g_1 f\,( D^2_{x} f \cdot  f+c (\rho_1^2 |g_1|^2+\rho_2^2 |g_2|^2)- c(\rho_1^2+\rho_2^2) f^2 ) =0\,,
\end{align}
where $D$ is the Hirota's bilinear operator \cite{hirota2004direct} defined by
\begin{eqnarray*}\label{doperator}
	D_x^mD_t^nf\cdot g=\left.\left(\frac{\partial}{\partial x}-\frac{\partial}{\partial {x'}}\right)^m\left(\frac{\partial}{\partial t}-\frac{\partial}{\partial{t'}}\right)^n
	[f(x,t)g(x',t')]\right|_{x'=x,t'=t}.
\end{eqnarray*}
If we require
\begin{align}\label{2}
	D^2_{x} f \cdot f+c (\rho_1^2 |g_1|^2+\rho_2^2 |g_2|^2) = c(\rho_1^2+\rho_2^2) f^2,
\end{align}
we obtained
$$
( D_x^3 - D_t + 3 \mathrm{i} \alpha_1 D_x^2 - 3 (\alpha_1^2+c(\rho_1^2+\rho_2^2)) D_x) g_1 \cdot f =0.  
$$
Similarly, the bilinear form for (\ref{ccmkdv2}) can also get. Thus, the resulting bilinear equations are
\begin{subequations}\label{bilinear-Hirota}
	\begin{align}
		&( D_x^3-D_t+3\mathrm{i} \alpha_1 D_x^2 - 3 (\alpha_1^2+c(\rho_1^2+\rho_2^2)) D_x)\,g_1 \cdot f =0\,, \tag{2.4a} \\
		&( D_x^3-D_t+3\mathrm{i} \alpha_2 D_x^2 - 3 (\alpha_2^2+c(\rho_1^2+\rho_2^2)) D_x)\,g_2 \cdot f =0\,, \tag{2.4b} \\
		&(D^2_x -c(\rho_1^2+\rho_2^2) )\,f \cdot f  = -c (\rho_1^2 |g_1|^2+\rho_2^2 |g_2|^2)\,. \tag{2.4c}
	\end{align}
\end{subequations}
If we introduce the following transformation, 
\begin{align}\label{3}
	\tilde{x}=x-3c(\rho_1^2+\rho_2^2)t,\ \ \tilde{t}=t, 
\end{align}
the above equations can be rewritten as 
\begin{subequations} \label{bilinear-Hirota2}
	\begin{align}
		&( D_x^3 - D_t + 3   \mathrm{i} \alpha_1 D_x^2 - 3 \alpha_1^2 D_x)\,\tilde{g}_1 \cdot \tilde{f} =0 \,,\tag{2.6a}\label{BHa} \\ 
		&( D_x^3 - D_t + 3  \mathrm{i} \alpha_2 D_x^2 - 3 \alpha_2^2 D_x)\,\tilde{g}_2 \cdot \tilde{f} =0 \,,\tag{2.6b} \label{BHb} \\ 
		& (D^2_x -c(\rho_1^2+\rho_2^2) )\,\tilde{f} \cdot \tilde{f}  = -c (\rho_1^2 |\tilde{g}_1|^2+\rho_2^2 |\tilde{g}_2|^2) \,,  \tag{2.6c} \label{BHc}
	\end{align} 
\end{subequations}
where $\tilde{f}=\tilde{f}(\tilde{x},\tilde{t}), \tilde{g_i}=\tilde{g_i}(\tilde{x},\tilde{t}), i=1,2.$ Therefore, the ccmKdV equation (\ref{ccmkdv1})--(\ref{ccmkdv2}) can be transformed into the bilinear equations \eqref{BHa}--\eqref{BHc} via the transformations (\ref{coupled HH}) and (\ref{3}). 

We convert the solution of system (\ref{ccmkdv1})--(\ref{ccmkdv2}) to the solution of 
\eqref{BHa}--\eqref{BHc}. 
In what follows, we construct solutions satisfying \eqref{BHa}--\eqref{BHc} for $c=1$, and they are expressed in pfaffian.
\begin{theorem}\label{thmmain}
	The dark-dark soliton solution of the  ccmKdV equation (\ref{ccmkdv1})--(\ref{ccmkdv2}) satisfies   
	\begin{equation} \label{dark-dark}
		u_1=\rho_1 \frac{g_1}{f} e^{\mathrm{i}\left(\alpha_1 x -(\alpha^3_1 + 3 c \alpha_1 (\rho_1^2+\rho_2^2) )t\right)}, \ \ u_2= \rho_2 \frac{g_2}{f} e^{\mathrm{i}\left(\alpha_2 x -(\alpha^3_2 + 3 c \alpha_2 (\rho_1^2+\rho_2^2) ) t\right)},
	\end{equation}
	where 
	$$
	f=\tau_{0,0} \left( x-3c(\rho_1^2+\rho_2^2)t, t \right), \ \ g_1=\tau_{1,0}  \left( x-3c(\rho_1^2+\rho_2^2)t, t \right), \ \ g_2=\tau_{0,1}  \left( x-3c(\rho_1^2+\rho_2^2)t, t \right),
	$$
	The pfaffian $\tau_{k_1,k_2}$ given by 
	\begin{equation}\label{tau_fun}
		\tau_{k_1,k_2}=(1,2,\cdots,2N)
	\end{equation} 
	with
	\begin{eqnarray}
		&& (i,j)=\delta_{2N+1-i,j}+\frac{p_i-p_j}{p_i+p_j} (d_0,i)(d_0,j), \ \ i<j\\
		&&(d_0,i)= \prod_{\nu=1}^2 \left(\frac{p_i-\mathrm{i} \alpha_\nu }{p_i+\mathrm{i} \alpha_\nu }\right)^{k_{\nu}}\exp \xi_i, \ \ 
		\xi_i=p_i x+p_i^3 t+\xi_{i 0} ,
	\end{eqnarray}
	the parameters $p_i$ and $\xi_{i 0}$ are complex constants and satisfying the reduction condition
	\begin{eqnarray}\label{Reduction_cond}
		\frac{2{\alpha_1}^2 {\rho_1}^2} {(p^2_i+\alpha^2_1 )({p^*_i}^2+\alpha^2_1)} + \frac{2{\alpha_2}^2 {\rho_2}^2}{(p^2_i+\alpha^2_2)({p^*_i}^2+\alpha^2_2)} =1,
	\end{eqnarray}
	where $ p_{2N+1-i}=p^*_i,$ and  $\xi_{2N+1-i} = \xi^*_{i} + \mathrm{i} \pi/2$ for $1 \leq i \leq N.$
\end{theorem}

\section{Derivation of the dark soliton solution}\label{Derivation of dark soliton Solutions of the coupled Hirota equation}
In this section, we will derive multi-dark soliton solution to the ccmKdV equation (\ref{ccmkdv1})--(\ref{ccmkdv2}) which satisfy three bilinear equations \eqref{BHa}--\eqref{BHc}. We start with a lemma, which can be deduced from the discrete BKP hierarchy. 

\begin{lemma}
	The pfaffian
	\begin{align}\label{tau_fun2}
		\tau_{k_1, k_2}=(1,2,\cdots,2N)\,,
	\end{align}
	with its elements defined by 
	\begin{eqnarray}
		(i,j)=c_{ij}+\frac{p_i-p_j}{p_i+p_j} (d_0,i) (d_0,j),\ \
		(d_0,i)= \prod_{\nu=1}^2 \left(\frac{p_i+a_\nu }{p_i-a_\nu }\right)^{k_{\nu}}\ \exp \xi_i \,,
	\end{eqnarray}
	where $\xi_i=p_i x_1+p_i^3 x_3+\xi_{i 0}$ and $c_{i j}=-c_{j i}$, $p_i$, $\xi_{i 0}$, $a_1$, $a_2$ are constants,
	satisfies the bilinear equations  
	\begin{align}
		&\left(D_{x_1}^{3}- D_{x_3}-3 a_1 D_{x_1}^{2}+3 a_1^{2} D_{x_1}\right) \tau_{k_1+1, k_2} \cdot \tau_{k_1,  k_2} =0\,, \label{BKP-3}\\
		&\left(D_{x_1}^{3}- D_{x_3}-3 a_2 D_{x_1}^{2}+3 a_2^{2} D_{x_1}\right) \tau_{k_1, k_2+1} \cdot \tau_{k_1,  k_2} =0 \,. \label{BKP-4}     
	\end{align}
\end{lemma}
\begin{proof}
	The discrete BKP equation was proposed by Miwa \cite{miwa1982hirota}.
	The pfaffian solution to the discrete BKP equation was given by Tsujimoto and Hirota \cite{tsujimoto1996pfaffian}.
	Similarly to the discrete KP equation, the discrete BKP equation can be extended to the discrete BKP hierarchy in the sense that the same bilinear equation holds from an arbitrary triple $(k_i,k_j,k_m)$ with discrete parameters $(a_i,a_j,a_m)$. Here we choose a pafffian with four discrete variables
	\begin{equation}
		\tau(k_1,k_2,k_3,k_4)=
		(1,2,\cdots, 2N)\,,
	\end{equation}
	with
	\begin{eqnarray}
		&&    
		(i,j)=c_{ij}+\frac{p_i-p_j}{p_i+p_j} \phi_i(k_1,k_2,k_3,k_4) \phi_j(k_1,k_2,k_3,k_4)\,, \\
		&& c_{i j}=-c_{j i},\ \ \phi_i(k_1,k_2,k_3,k_4)=\prod_{\nu=1}^4 \left(\frac{1-a^{-1}_\nu p_i}{1+a^{-1}_\nu p_i}\right)^{-k_{\nu}}\,.
	\end{eqnarray}
	If we pick up a triple ($k_1,k_3,k_4$)
	then we have 
	\begin{eqnarray}
		\label{dBKP2}\nonumber
		&&(a_{3}-a_{4}) (a_{1}+a_{3}) (a_{1}+a_{4})\tau _{1} {\tau }_{34}
		+(a_{4}-a_{1})(a_{3}+a_{4})(a_{3}+a_{1})\tau _{3}{\tau }_{41} \\
		&&+(a_{1}-a_{3}) (a_{4}+a_{1}) (a_{4}+a_{3}) \tau _{4}{\tau }_{13}+ (a_{1}-a_{3}) (a_{3}-a_{4}) (a_{4}-a_{1}) \tau {\tau }_{134} =0\,.
	\end{eqnarray}
	Here each subscript $i$ denotes a forward shift in the corresponding discrete variable~$n_i$, for example, $\tau_{i} = \tau(k_i+1,k_j,k_m)$, $\tau_{ij} = \tau(k_i+1,k_j+1,k_m)$.
	Notice that 
	\begin{eqnarray*}
		\left( \frac{(1-a^{-1}_{l}p_{i}) (1-a^{-1}_{l}p_{j})}{(1+a^{-1}_{l}p_{i}) (1+a^{-1}_{l}p_{j})}\right) ^{-k_{l}} &=&\exp \left(
		-k_{l}\ln \left( \frac{(1-a^{-1}_{l}p_{i}) (1-a^{-1}_{l}p_{j})}{(1+a^{-1}_{l}p_{i}) (1+a^{-1}_{l}p_{j})}\right) \right) \\
		&=&\exp \left( 2a^{-1}_{l}k_{l}(p_{i}+p_{j})+\frac{2}{3}
		a_{l}^{-3}k_{l} (p_{i}^{3}+p_{j}^{3})+\cdots \right)\,,
	\end{eqnarray*}
	we can define the so-called Miwa transformation
	\begin{equation*}
		x_1=2\sum^4_{i=3}k_i a^{-1}_i,\ x_3=\frac{2}{3}\sum^4_{i=3}k_ia_i^{-3},\ \cdots,\  x_{2\mu-1}=\frac{2}{2\mu-1}\sum^4_{i=3}k_ia_i^{-2\mu+1}.
	\end{equation*}
	Furthermore, we can define the elementary Schur polynomial
	\[
	\exp \left( \sum t^{n}x_{n}\right) =\sum p_{n}(\vec{x})t^{n}\,,
	\]%
	where
	\[
	\vec{x}=(x_{1},x_{2},\cdots ,x_{n})\,,
	\]
	\[
	p_{0}=1,\ p_{1}(x)=x_{1},\ p_{2}=x_{2}+\frac{1}{2}%
	x_{1}^{2},\cdots,  p_{3}=x_{3}+x_{1}x_{2}+\frac{1}{6}x_{1}^{3},
	\]
	\[
	p_{n}= \sum_{k_1+2k_2+\cdots+nk_n=n} \frac{x_1^{k_1}x_2^{k_2} \cdots x_n^{k_n}}{k_1! k_2! \cdots k_n!}\,.
	\]
	Then, the discrete BKP equation \eqref{dBKP2} can be converted into
	\begin{eqnarray*}
		&&((a_{3}-a_{4}) (a_{1}+a_{3}) (a_{1}+a_{4}) \sum a_{3}^{-L}a_{4}^{-M} p_{L}( -\widetilde{D}) p_{M}(-
		\widetilde{D}) \\
		&&+ (a_{4}-a_{1})(a_{3}+a_{4})(a_{3}+a_{1}) \sum a_{3}^{-L}a_{4}^{-M} p_{L}( -\widetilde{D})
		p_{M}( \widetilde{D})  \\
		&&+ (a_{1}-a_{3}) (a_{4}+a_{1}) (a_{4}+a_{3}) \sum a_{3}^{-L}a_{4}^{-M} p_{L}( \widetilde{D}) p_{M}(-\widetilde{D}) \\
		&& + (a_{1}-a_{3}) (a_{3}-a_{4}) (a_{4}-a_{1}) \sum a_{3}^{-L}a_{4}^{-M} p_{L}( \widetilde{D}) p_{M}(\widetilde{D})
		) \tau_{k_1+1,k_2} \cdot \tau_{k_1,k_2}=0\,.
	\end{eqnarray*}
	Here
	$
	\widetilde{D}=(D_{x_1}, 0, \frac 13 D_{x_3}, 0, \cdots)$. 
	At the order of $a^{-1}_3a_4$, we have
	\begin{eqnarray*}
		&& (4p_3(\widetilde{D})  -4 p_1(\widetilde{D}) p_2(\widetilde{D}) +4a_1 p^2_1(\widetilde{D}) -4a^2_1 p_1(\widetilde{D})) \tau_{k_1+1,k_2} \cdot \tau_{k_1,k_2}=0\,,
	\end{eqnarray*}
	which is nothing but the bilinear equation \eqref{BKP-3}. Similarly,  starting from 
	\begin{eqnarray}
		\label{dBKP22}
		&&(a_{3}-a_{4}) (a_{2}+a_{3}) (a_{2}+a_{4})\tau _{2} {\tau }_{34} \nonumber
		+(a_{4}-a_{2})(a_{3}+a_{4})(a_{3}+a_{2})\tau _{3}{\tau }_{42} \\
		&&+(a_{2}-a_{3}) (a_{4}+a_{2}) (a_{4}+a_{3}) \tau _{4}{\tau }_{23}+ (a_{2}-a_{3}) (a_{3}-a_{4}) (a_{4}-a_{2}) \tau {\tau }_{234} =0\,,
	\end{eqnarray} 
	by picking up a triple of ($k_2,k_3,k_4$), we can approve \eqref{BKP-4}. 
\end{proof}

{\bf Remark:} The same bilinear equations (\ref{BHa}) and (\ref{BHb}) were approved by the identities of pfaffian in \cite{feng2015integrable}.
The lemma given below is approved in Appendix by pfaffian identities. 
\begin{lemma}\label{lemma_two}
	If $c_{ij} = \delta_{2N+1-i,j}$, $1 \leq i<j \leq 2N$, then under the condition 
	\begin{eqnarray}
		\frac{2{a_1}^2 {\rho_1}^2} {(p^2_{j}-a^2_1)(p^2_{2N+1-j}-a^2_1)} +\frac{2{a_2}^2 {\rho_2}^2} {(p^2_j-a^2_2)(p^2_{2N+1-j}-a^2_2)} =-1\,,
	\end{eqnarray}
	the pfaffian (\ref{tau_fun2}) satisfies the following bilinear equation
	\begin{align}
		(D^2_x -(\rho_1^2+\rho_2^2)) \tau_{k_1,k_2} \cdot \tau_{k_1,k_2} = -(\rho_1^2 \tau_{k_1+1,k_2} \tau_{k_1-1,k_2}+\rho_2^2 \tau_{k_1,k_2+1} \tau_{k_1,k_2-1})\,. {\label{TODA}}
	\end{align}   
	
\end{lemma}

Proof of the Lemma is given in Appendix. We note that the equations in the BKP hierarchy exhibit a structure akin to that presented in (\ref{BHa})-(\ref{BHc}), albeit with supplementary conditions. This insight prompts us to explore the complex conjugate reduction.

\begin{lemma}
	If we choose $a_1= - \mathrm{i} \alpha_1$, $a_2= -\mathrm{i} \alpha_2$ to be purely imaginary numbers,  $c_{ij} = \delta_{2N+1-i,j}$ for $1 \leq i<j \leq 2N$, and $ p_{2N+1-i}=p^*_i,$ $\xi_{2N+1-i} = \xi^*_{i} + \mathrm{i} \pi/2$ for $1 \leq i \leq N $ in pfaffian (\ref{tau_fun2}), it it shown that
	\begin{align}
		\overline{\tau}_{k_1,k_2}=\tau_{-k_1,-k_2}.
	\end{align}
\end{lemma}
\begin{proof}
	Using the above conditions, we find
	\begin{align*}
		\overline{(d_0,i)} 
		& =-\mathrm{i} \prod_{\nu=1}^2 \left(\frac{p_{2N+1-i}-\mathrm{i} \alpha_\nu }{p_{2N+1-i}+\mathrm{i} \alpha_\nu }\right)^{-k_{\nu}}e^{\xi_{2N+1-i}} 
		= -\mathrm{i} (d_0,2N+1-i)_{-k_1,-k_2}\,, \ \ \ \ \ \ 1 \leq i \leq 2N\\
		\overline{(i,j)}
		&=- \delta_{2N+1-i,j}-\dfrac{p_{2N+1-i}-p_{2N+1-j}}{p_{2N+1-i}+p_{2N+1-j}} (d_0,2N+1-i)_{-k_1,-k_2} (d_0,2N+1-j)_{-k_1,-k_2} \\
		&=-(2N+1-i,2N+1-j)_{-k_1,-k_2}\\
		&=(2N+1-j,2N+1-i)_{-k_1,-k_2} \,,\ \ \ \ 1 \leq i< j \leq 2N
	\end{align*}
	Therefore,
	\begin{align*}
		\overline{\tau}_{k_1,k_2} &=((2N+1-j,2N+1-i)_{-k_1,-k_2})_{1 \leq i< j \leq 2N}=(1,2,\ldots,2N)_{-k_1,-k_2}=\tau_{-k_1,-k_2}.
	\end{align*} 
\end{proof}

Finally, we define 
\begin{align}
	\tilde{f}=\tau_{00}, \quad \tilde{g}_1=\tau_{1,0}, \quad \tilde{g}_2=\tau_{0,1},
\end{align} 
then we have
\begin{align}
	\tilde{f}=\tilde{f}^*, \quad 
	\tilde{{g}}^*_1=\tau_{-1,0}, \quad \tilde{{g}}^*_2=\tau_{0,-1}.
\end{align}

Combining Lemma 3.1-- Lemma 3.3, the pfaffian solutions to the bilinear equations (\ref{BHa})-(\ref{BHc}) are derived.

\section{Dynamics of one- and two-soliton {\color{red}}solutions}

Theorem \ref{thmmain} gives the dark-dark soliton solution (\ref{dark-dark}) of the ccmKdV equation (\ref{ccmkdv1})-(\ref{ccmkdv2}).
By taking $N=1$, we obtain the one-soliton solution in which $f,g_1,g_2$ are
\begin{align*}
	f&=1+\mathrm{i}\; \dfrac{p_1-p^*_1}{p_1+p^*_1} 
	\exp {(\xi_1+\xi^*_1)}\,,\\
	g_1&=1+\mathrm{i}\; \dfrac{p_1-p^*_1}{p_1+p^*_1} \left(\dfrac{p_1-\mathrm{i}\alpha_1}{p_1+\mathrm{i}\alpha_1}\right)\left(\dfrac{p^*_1-\mathrm{i}\alpha_1}{p^*_1+\mathrm{i}\alpha_1}\right)
	\exp {(\xi_1+\xi^*_1)}\,,\\
	g_2&=1+\mathrm{i}\; \dfrac{p_1-p^*_1}{p_1+p^*_1} \left(\dfrac{p_1-\mathrm{i}\alpha_2}{p_1+\mathrm{i}\alpha_2}\right)\left(\dfrac{p^*_1-\mathrm{i}\alpha_2}{p^*_1+\mathrm{i}\alpha_2}\right)
	\exp {(\xi_1+\xi^*_1)}\,,
\end{align*}
where the $\xi_1=p_1 (x-3(\rho^2_1+\rho
^2_2)t)+p^3_1 t+ \xi_{10}\,,$ and $p_1$ satisfy the reduction relation. The one-soliton solution is expressed as
\begin{align*}\nonumber
	u_1= \rho_1 \exp(\mathrm{i} \theta_1) &\left[ \dfrac{1}{2} \left(1+ \left(\dfrac{p_1-\mathrm{i}\alpha_1}{p_1+\mathrm{i}\alpha_1}\right)\left(\dfrac{p^*_1-\mathrm{i}\alpha_1}{p^*_1+\mathrm{i}\alpha_1}\right) \right) \right.\\
	&\left. - \dfrac{1}{2} \left(1- \left(\dfrac{p_1-\mathrm{i}\alpha_1}{p_1+\mathrm{i}\alpha_1}\right) \left(\dfrac{p^*_1-\mathrm{i}\alpha_1}{p^*_1+\mathrm{i}\alpha_1}\right) \right) \tanh \left(\mathrm{Re}(\xi_1)+\dfrac{1}{2} \mathrm{log} (\mathrm{i} \dfrac{\mathrm{Im}(p_1)}{\mathrm{Re}(p_1)})\right) \right],\\
	u_2= \rho_2 \exp(\mathrm{i} \theta_2) &\left[ \dfrac{1}{2} \left(1+ \left(\dfrac{p_1-\mathrm{i}\alpha_2}{p_1+\mathrm{i}\alpha_2}\right)\left(\dfrac{p^*_1-\mathrm{i}\alpha_2}{p^*_1+\mathrm{i}\alpha_2}\right) \right) \right.\\
	&\left. - \dfrac{1}{2} \left(1- \left(\dfrac{p_1-\mathrm{i}\alpha_2}{p_1+\mathrm{i}\alpha_2}\right) \left(\dfrac{p^*_1-\mathrm{i}\alpha_2}{p^*_1+\mathrm{i}\alpha_2}\right) \right) \tanh \left(\mathrm{Re}(\xi_1)+\dfrac{1}{2} \mathrm{log} (\mathrm{i} \dfrac{\mathrm{Im}(p_1)}{\mathrm{Re}(p_1)})\right) \right]\,,
\end{align*}
where the $\theta_i=\alpha_i x-(\alpha^3_i+3\alpha_i (\rho^2_1+\rho^2_2)) t,$ $i=1,2.$ We can plot the one dark soliton solutions, see Fig. 1. 

\begin{figure}    \label{fig:p1}
	\centering
	\subfigure[]{
		\includegraphics[width=0.45\textwidth]{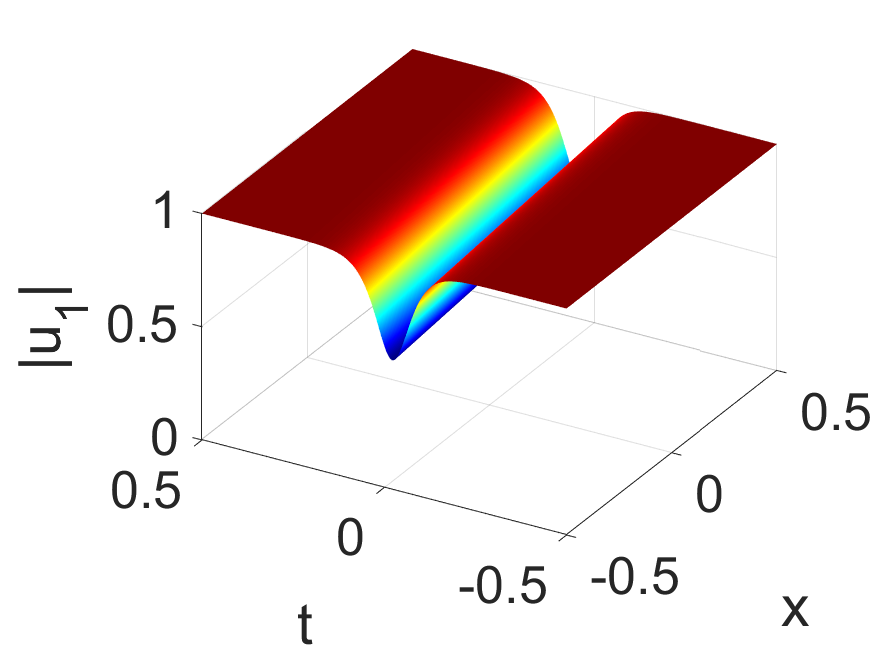}
	}
	\subfigure[]{
		\includegraphics[width=0.45\textwidth]{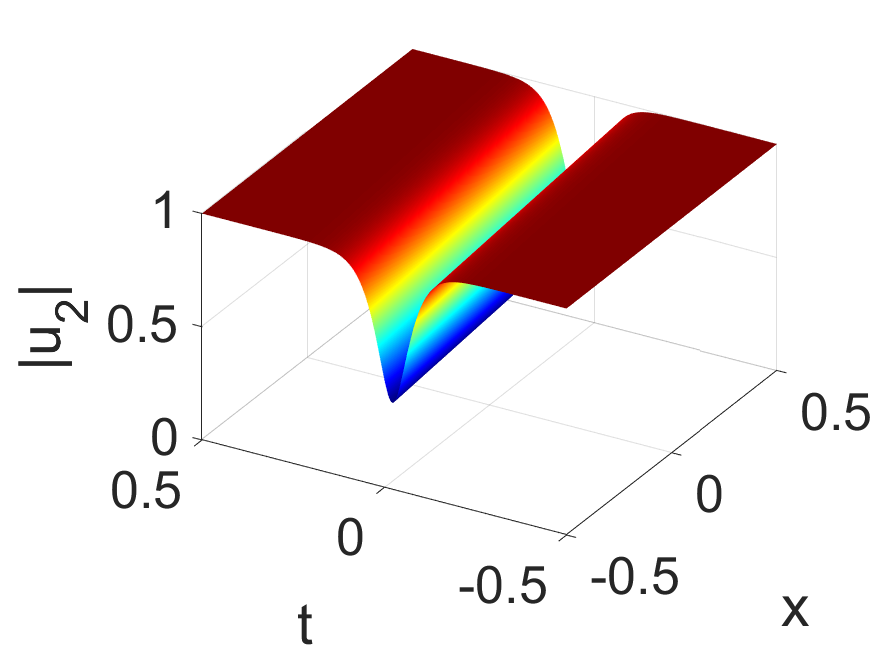}
	}
	\caption{ One dark soliton solutions to the ccmKdV equation (\ref{ccmkdv1})--(\ref{ccmkdv2}) with parameters  $p_1$ = 0.88 + i, $\rho1 = 1,\,\alpha1 = 2$,\,$\rho2 = 1,\, \alpha2 = 1.$ (a) and (b) are  profiles of  $|u_1|$ and $|u_2|$, respectively.} 
\end{figure}
Taking $N=2$, we get the two-soliton solution
\begin{align*}
	f&=1+C_1 
	\exp {(\xi_1+\xi^*_1)}+C_2  
	\exp {(\xi_2+\xi^*_2)}+C_1 C_2 \mathrm{M} \exp {(\xi_1+\xi^*_1+\xi_2+\xi^*_2)}\,,\\
	g_1&=1+C_1 A_1 
	\exp {(\xi_1+\xi^*_1)}+C_2 A_2
	\exp {(\xi_2+\xi^*_2)}+C_1 C_2 A_1 A_2 \mathrm{M} \exp {(\xi_1+\xi^*_1+\xi_2+\xi^*_2)}\,,\\
	g_2&=1+C_1 B_1
	\exp {(\xi_1+\xi^*_1)}+C_2 B_2
	\exp {(\xi_2+\xi^*_2)}+C_1 C_2 B_1 B_2 \mathrm{M} \exp {(\xi_1+\xi^*_1+\xi_2+\xi^*_2)}\,,
\end{align*}
and
\begin{align*}
	C_j=\mathrm{i}\,\dfrac{p_j-p^*_j}{p_j+p^*_j}&,\ \ A_j=\dfrac{(p_j-\mathrm{i}\alpha_1)(p^*_j-\mathrm{i}\alpha_1)}{(p_j+\mathrm{i}\alpha_1)(p^*_j+\mathrm{i}\alpha_1)},\ \ B_j=\dfrac{(p_j-\mathrm{i}\alpha_2)(p^*_j-\mathrm{i}\alpha_2)}{(p_j+\mathrm{i}\alpha_2)(p^*_j+\mathrm{i}\alpha_2)},\\
	&\mathrm{M}=\dfrac{(p_1-p_2)(p^*_1-p^*_2)(p_1-p^*_2)(p^*_1-p_2)}{(p_1+p_2)(p^*_1+p^*_2)(p_1+p^*_2)(p^*_1+p_2)},
\end{align*}
where the $\xi_i=p_i (x-3(\rho^2_1+\rho^2_2)t)+p^3_i t+ \xi_{i0},\, i,j =1,2$ and $p_i$ satisfy the complex relation.
One can observe the illustration of two-soliton solution in Fig. 2. Next, we consider the dynamics behavior of above solution \cite{feng2018general}. The one with $ 
\xi_1+\xi^*_1$ is called soliton 1 and another $ 
\xi_2+\xi^*_2$ is called soliton 2.  Assume soliton 1 is on the right of soliton 2 before the collision. Hence,

\begin{figure}\label{fig:p2}
	\centering
	\subfigure[]{
		\includegraphics[width=0.45\textwidth]{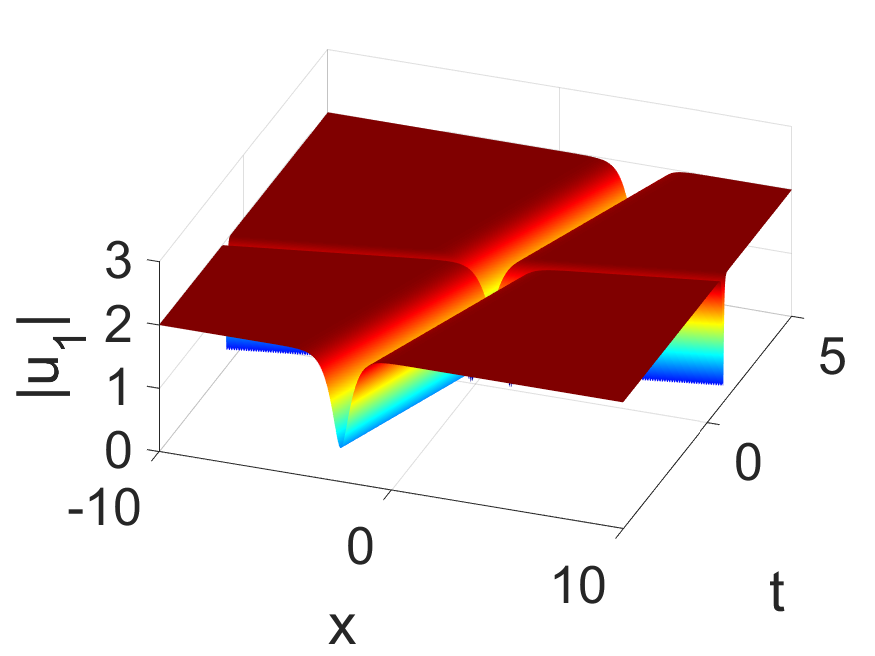}
	}
	\subfigure[]{
		\includegraphics[width=0.45\textwidth]{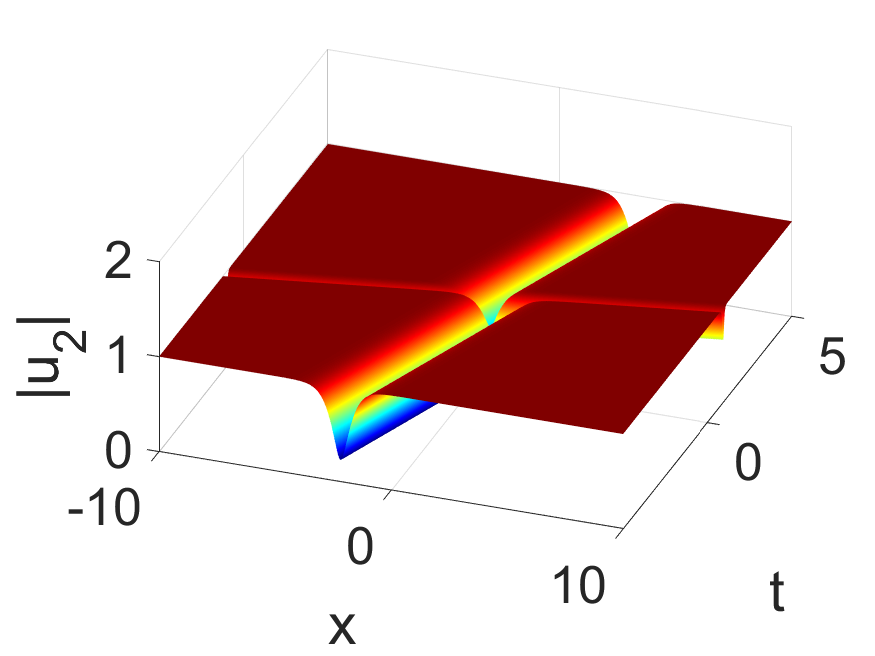}
	}
	\subfigure[]{
		\includegraphics[width=0.45\textwidth]{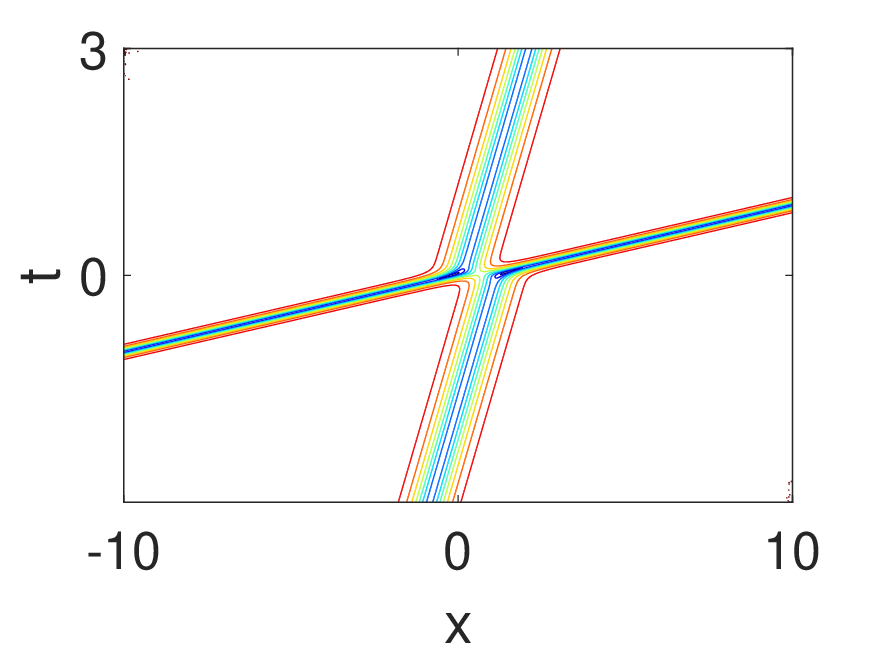}
	} 
	\subfigure[]{
		\includegraphics[width=0.45\textwidth]{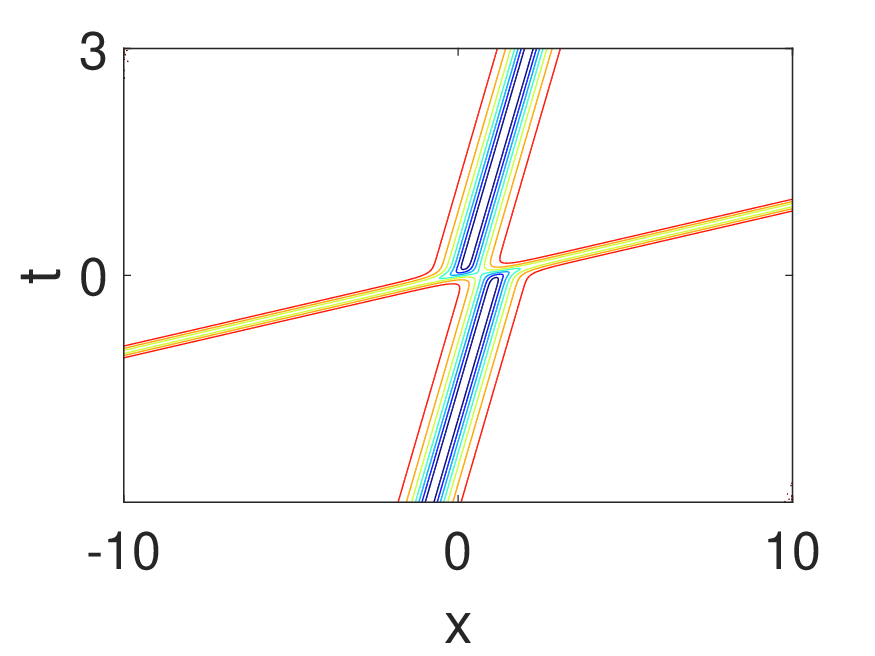}
	}
	\caption{ Two dark soliton solutions to the ccmKdV equation (\ref{ccmkdv1})--(\ref{ccmkdv2}) with parameters $p_1$ = 1.53 + i, $p_2$ = 1.49 + 2i, $\rho_1 = 2,\alpha_1 = 2.3$, $\rho_2 = 1, \alpha_2 = 1.5$, (a) the profile of  $|u_1|$, (b) the profile of $|u_2|$, (c) and (d) are counter plots of $|u_1|$ and $|u_2|$ respectively.}
\end{figure}

(1) Before collision, i.e., $t \rightarrow -\infty$\\
Soliton 1 ($\xi_1+\xi^*_1 \approx 0,\ \ \xi_2+\xi^*_2 \rightarrow -\infty$)
\begin{align*}
	u_1 \rightarrow \rho_1 \exp(\mathrm{i} \theta_1) \left[  \dfrac{1}{2} \left(1+A_1 \right)- \dfrac{1}{2} \left(1-A_1 \right) \tanh \left(\mathrm{Re}(\xi_1)+\dfrac{1}{2} \mathrm{log} \left( \chi_1 \right)\right) \right],\\
	u_2 \rightarrow \rho_2 \exp(\mathrm{i} \theta_2) \left[  \dfrac{1}{2} \left(1+B_1 \right)- \dfrac{1}{2} \left(1-B_1 \right) \tanh \left(\mathrm{Re}(\xi_1)+\dfrac{1}{2} \mathrm{log} \left(\chi_1\right)\right) \right].
\end{align*}
Soliton 2 ($\xi_2+\xi^*_2 \approx 0,\ \ \xi_1+\xi^*_1 \rightarrow +\infty$)
\begin{align*}
	u_1 \rightarrow  \rho_1  A_1 \mathrm{exp} (\mathrm{i} \theta_1) \left[\dfrac{1}{2} \left(1+A_2  \right) - \dfrac{1}{2} \left(1-A_2 \right) \tanh \left(\mathrm{Re}(\xi_2)+\dfrac{1}{2} \mathrm{log} \left( \gamma_1 \right)\right) \right],\\
	u_2 \rightarrow \rho_2  B_1 \mathrm{exp} (\mathrm{i} \theta_2) \left[\dfrac{1}{2}  \left(1+B_2 \right) - \dfrac{1}{2}  \left(1-B_2 \right) \tanh  \left(\mathrm{Re}(\xi_2)+\dfrac{1}{2} \mathrm{log}  \left(\gamma_1 \right)\right) \right].
\end{align*}
where the phases term in Soliton 1 and Soliton 2 was respectively expressed as 
$$\chi_1=\mathrm{i} \dfrac{\mathrm{Im}(p_1)}{\mathrm{Re}(p_1)}, \ \ 
\gamma_1=\mathrm{i} \dfrac{\mathrm{Im}(p_2)}{\mathrm{Re}(p_2)} M . $$
(2) After collision, i.e., $t \rightarrow +\infty$\\
Soliton 1 ($\xi_1+\xi^*_1 \approx 0,\ \ \xi_2+\xi^*_2 \rightarrow +\infty$)
\begin{align*}
	u_1 \rightarrow \rho_1  A_2  \mathrm{exp} (\mathrm{i} \theta_1) \left[\dfrac{1}{2} \left(1+A_1 \right) - \dfrac{1}{2} \left(1-A_1 \right) \tanh \left(\mathrm{Re}(\xi_1)+\dfrac{1}{2} \mathrm{log} \left(\chi_2 \right)\right) \right],\\
	u_2 \rightarrow \rho_2  B_2  \mathrm{exp} (\mathrm{i} \theta_2) \left[\dfrac{1}{2} \left(1+B_1 \right) - \dfrac{1}{2} \left(1-B_1 \right) \tanh \left(\mathrm{Re}(\xi_1)+\dfrac{1}{2} \mathrm{log} \left(\chi_2 \right)\right) \right].
\end{align*}
Soliton 2 ($\xi_2+\xi^*_2 \approx 0,\ \ \xi_1+\xi^*_1 \rightarrow -\infty$)
\begin{align*}
	u_1 \rightarrow \rho_1 \exp(\mathrm{i} \theta_1) \left[  \dfrac{1}{2} \left(1+A_2 \right)- \dfrac{1}{2} \left(1-A_2 \right) \tanh \left(\mathrm{Re}(\xi_2)+\dfrac{1}{2} \mathrm{log} \left( \gamma_2 \right) \right) \right],\\
	u_2 \rightarrow \rho_2 \exp(\mathrm{i} \theta_2) \left[  \dfrac{1}{2} \left(1+B_2 \right)- \dfrac{1}{2} \left(1-B_2\right) \tanh \left(\mathrm{Re}(\xi_2)+\dfrac{1}{2} \mathrm{log} \left(\gamma_2 \right)\right) \right].
\end{align*}
where the phase term in Soliton 1 and Soliton 2 was respectively expressed as 
$$\chi_2=\mathrm{i} \dfrac{\mathrm{Im}(p_1)}{\mathrm{Re}(p_1)} M, \ \  \gamma_2=\mathrm{i} \dfrac{\mathrm{Im}(p_2)}{\mathrm{Re}(p_2)}. $$
In Fig. 2, it can be found that a phase shift occurs after the solitons collide with each other due to the difference in phase terms $\chi_i$ and $\gamma_i$. From the contour plot in Fig. 2, it is observed that the intensity of the dark solitons does not decrease after the collision.

\section{Conclusion}
In this paper, we derived the general dark-dark soliton solutions in the form of pfaffian for the ccmKdV equation (\ref{ccmkdv1})--(\ref{ccmkdv2}) under nonzero boundary condition. This work solved an open problem existed for a long time. The crucial step is to link the two bilinear equations of the ccmKdV equaiton to the disrete BKP hierarchy through Miwa transformation. As future work, we will investigate the soliton solutions to the semi-discrete version of the ccmKdV equation and report the results elsewhere.

\section{Appendix: Proof of the lemma 3.2}
\begin{proof}
	By defining
	\begin{align}
		(d_m,i)=p^{m}_i \prod_{\nu=1}^2 \left(\frac{p_i+a_\nu }{p_i-a_\nu }\right)^{k_{\nu}}\ \exp \xi_i \,,
	\end{align}
	it is shown that the pfaffian elements in (\ref{tau_fun2}) satisfy the linear relations
	\begin{align}
		\dfrac{\partial}{\partial x}(i,j)&=(d_1,i)(d_0,j)-(d_0,i)(d_1,j),\label{proof_re_1}\\
		(i,j)_{k_1+1}-(i,j)&=(d_0,j)_{k_1+1} (d_0,i)-(d_0,j)(d_0,i)_{k_1+1}.\label{proof_re_2}
	\end{align}
	Here we omit without changing terms, for example,
	$$(i,j)_{k_1+1}=(i,j)_{k_1+1,k_2},\ \ (i,j)=(i,j)_{k_1,k_2},\ \ (d_0,i)=(d_0,i)_{k_1,k_2}.$$
	We can also show 
	\begin{eqnarray}
		&&\tau_{k_1+1,k_2}=(\bar{d_0},d_0,1,2,\cdots,2N),\\
		&&\tau_{k_1-1,k_2}=(\tilde{d_0},d_0,1,2,\cdots,2N),
	\end{eqnarray}
	where  $(\bar{d_0},k)=(d_0,k)_{k_1+1}$, $(\tilde{d_0},k)=(d_0,k)_{{k_1}-1}$ and $(\bar{d_0},d_0)=(\tilde{d_0},d_0)=1.$
	Thus, we have
	\begin{align}
		\tau_{k_1+1,k_2}=\tau_{k_1,k_2}+\sum_{j=1}^{2N} (-1)^{j} (d_0,1,\cdots,\hat{j},\cdots,2N)(d_0,j)_{k_1+1}\label{k_1+1}\,,\\
		\tau_{k_1-1,k_2}=\tau_{k_1,k_2}+\sum_{j=1}^{2N} (-1)^{j} (d_0,1,\cdots,\hat{j},\cdots,2N)(d_0,j)_{k_1-1}\,.   \label{k_1-1}
	\end{align}
	Utilizing the (\ref{k_1+1}) and (\ref{k_1-1}), we can arrive at
	\begin{eqnarray*}
		&&\tau_{k_1+1,k_2}\tau_{k_1-1,k_2}-\tau_{k_1,k_2}\tau_{k_1,k_2}\\
		&&= \tau_{k_1,k_2} \sum_{j=1}^{2N}(-1)^{j} (d_0,1,\cdots,\hat{j},\cdots,2N)[(d_0,j)_{k_1+1}+(d_0,j)_{k_1-1}]\\
		&&+\sum_{{1\le i< j\le 2N}}(-1)^{i+j}(d_0,1,\cdots,\hat{i},\cdots,2N)(d_0,1,\cdots,\hat{j},\cdots,2N)\\
		&&\times [(d_0,i)_{k_1+1}(d_0,j)_{k_1-1}+(d_0,i)_{k_1-1}(d_0,j)_{k_1+1}]\\
		&&+\sum_{j=1}^{2N}(d_0,1,\cdots,\hat{j},\cdots,2N)^2 (d_0,j)_{k_1+1}(d_0,j)_{k_1-1}\\   
		&&=\tau_{k_1,k_2} \sum_{j=1}^{2N}(-1)^{j} (d_0,1,\cdots,\hat{j},\cdots,2N)[(d_0,j)_{k_1+1}+(d_0,j)_{k_1-1}]  \\
		&&+\sum_{{1\le i< j\le 2N}}(-1)^{i+j}(d_0,1,\cdots,\hat{i},\cdots,2N)(d_0,1,\cdots,\hat{j},\cdots,2N)\\
		&&\times [(d_0,j)_{k_1+1}(d_0,k)_{k_1-1}+(d_0,j)_{k_1-1}(d_0,k)_{k_1+1}-2(d_0,j)(d_0,k)]\\
		&&+\sum_{j=1}^{2N}[(d_0,1,\cdots,\hat{j},\cdots,2N)(d_0,j)]^2
	\end{eqnarray*}
	Here this identity $(d_0,j)_{k_1+1}(d_0,j)_{k_1-1}=(d_0,j)^2$ is used. Now, the last two items can be simplified. By using
	\begin{eqnarray}\nonumber
		&&(d_0,i)_{k_1+1}(d_0,j)_{k_1-1}+(d_0,i)_{k_1-1}(d_0,j)_{k_1+1}-2(d_0,i)(d_0,j)\\
		&&=\left(\dfrac{p_j+a_1}{p_j-a_1}+\dfrac{p_j-a_1}{p_j+a_1}-\dfrac{p_i+a_1}{p_i-a_1}-\dfrac{p_i-a_1}{p_i+a_1}\right)((i,j)-\delta_{2N+1-i,j}),\ \ i<j.
	\end{eqnarray}
	So, we have
	\begin{eqnarray*}
		&&\tau_{k_1+1,k_2}\tau_{k_1-1,k_2}-\tau_{k_1,k_2}\tau_{k_1,k_2}\\
		&&= \tau_{k_1,k_2} \sum_{j=1}^{2N}(-1)^{j} (d_0,1,\cdots,\hat{j},\cdots,2N)[(d_0,j)_{k_1+1}+(d_0,j)_{k_1-1}]\\
		&&+\sum_{{1\le i< j\le 2N}}(-1)^{i+j}(d_0,1,\cdots,\hat{i},\cdots,2N)(d_0,1,\cdots,\hat{j},\cdots,2N)\\
		&&\times \left(\dfrac{p_j+a_1}{p_j-a_1}+\dfrac{p_j-a_1}{p_j+a_1}-\dfrac{p_i+a_1}{p_i-a_1}-\dfrac{p_i-a_1}{p_i+a_1}\right)((i,j)-\delta_{2N+1-i,j})\\
		&&=\tau_{k_1,k_2} \sum_{j=1}^{2N}(-1)^{j} (d_0,1,\cdots,\hat{j},\cdots,2N)\left(\dfrac{p_j+a_1}{p_j-a_1}+\dfrac{p_j-a_1}{p_j+a_1}\right)(d_0,j)\\
		&&+\sum_{i=1}^{2N}\sum_{j=1}^{2N}(-1)^{i+j-1}(d_0,1,\cdots,\hat{i},\cdots,2N)(d_0,1,\cdots,\hat{j},\cdots,2N) \left(\dfrac{p_j-a_1}{p_j+a_1}+\dfrac{p_j+a_1}{p_j-a_1}\right)(i,j)\\
		&&+\sum_{j=1}^N (d_0,1,\cdots,\hat{j},\cdots,2N)(d_0,1,\cdots,\widehat{2N+1-j},\cdots,2N)\\
		&&\times \left(\dfrac{{p_{2N+1-j}}+a_1}{{p_{2N+1-j}}-a_1}+\dfrac{{p_{2N+1-j}}-a_1}{{p_{2N+1-j}}+a_1}-\dfrac{p_j+a_1}{p_j-a_1}-\dfrac{p_j-a_1}{p_j+a_1}\right)
	\end{eqnarray*}
	and using
	\begin{align}
		\tau_{k_1,k_2} (d_0,i)+\sum_{j=1}^{2N}(-1)^{j-1}(d_0,1,\cdots,\hat{j},\cdots,2N)(i,j)=(d_0,i,1,\cdots,2N)=0,\label{L}
	\end{align} 
	the sum of the first two terms vanished. Therefore, 
	\begin{align}\nonumber
		&\tau_{k_1+1,k_2}\tau_{k_1-1,k_2}-\tau_{k_1,k_2}\tau_{k_1,k_2}\\
		&=\sum_{j=1}^N \dfrac{4a^2_1 (p_{j}^2-{p^2_{2N+1-j}})}{(p_{j}^2-a^2_1)({p^2_{2N+1-j}}-a^2_1)}(d_0,1,\cdots,\hat{j},\cdots,2N)(d_0,1,\cdots,\widehat{2N+1-j},\cdots,2N).
	\end{align}
	Similarly, we have
	\begin{align}\nonumber
		&\tau_{k_1,k_2+1}\tau_{k_1,k_2-1}-\tau_{k_1,k_2}\tau_{k_1,k_2}\\
		&=\sum_{j=1}^N \dfrac{4a^2_2 (p_{j}^2-{p^2_{2N+1-j}})}{(p_{j}^2-a^2_2)({p^2_{2N+1-j}}-a^2_2)}(d_0,1,\cdots,\hat{j},\cdots,2N)(d_0,1,\cdots,\widehat{2N+1-j},\cdots,2N).
	\end{align}
	By applying (\ref{proof_re_1}), we can further derive
	\begin{eqnarray}
		&&\dfrac{\partial}{\partial x} \tau_{k_1,k_2}
		=(d_0,d_1,1,2,\cdots,2N), \ \ \dfrac{\partial^2}{\partial x^2} \tau_{k_1,k_2}=(d_0,d_2,1,2,\cdots,2N), 
	\end{eqnarray}
	where the pfaffian $(d_0,d_1)=(d_0,d_2)=0$. Then
	\begin{flalign}
		\dfrac{\partial}{\partial x} \tau_{k_1,k_2}
		=\sum_{j=1}^{2N}(-1)^{j-1}(d_0,1,\cdots,\hat{j},\cdots,2N)(d_1,j),\label{partial_tau1} \\
		\dfrac{\partial^2}{\partial x^2} \tau_{k_1,k_2}
		=\sum_{j=1}^{2N}(-1)^{j-1}(d_0,1,\cdots,\hat{j},\cdots,2N)(d_2,j).\label{partial_tau2}
	\end{flalign}
	Through formulas (\ref{partial_tau1}) and (\ref{partial_tau2}), we can achieve
	\begin{eqnarray}\nonumber
		&&(\partial_{x}^2 \tau_{k_1,k_2})\tau_{k_1,k_2}-(\partial_x \tau_{k_1,k_2})^2\\\nonumber
		&&=\tau_{k_1,k_2} [ \sum_{j=1}^{2N}(-1)^{j-1}(d_0,1,\cdots,\hat{j},\cdots,2N)(d_2,j) ]-[\sum_{j=1}^{2N}(-1)^{j-1}(d_0,1,\cdots,\hat{j},\cdots,2N)(d_1,j)]^2\\\nonumber
		&&=\tau_{k_1,k_2}\sum_{j=1}^{2N}(-1)^{j-1} p_{j}^2 (d_0,1,\cdots,\hat{j},\cdots,2N)(d_0,j)-[\sum_{j=1}^{2N}(-1)^{j-1}(d_0,1,\cdots,\hat{j},\cdots,2N)(d_1,j)]^2
	\end{eqnarray}
	the first term in the aforementioned equation can be simplified by applying (\ref{L}),
	\begin{eqnarray*}\nonumber
		&&=\sum_{i=1}^{2N}\sum_{j=1}^{2N}(-1)^{i+j-1} p_{i}^2 (d_0,1,\cdots,\hat{i},\cdots,2N)(d_0,1,\cdots,\hat{j},\cdots,2N)(i,j)\\\nonumber
		&&-[\sum_{j=1}^{2N}(-1)^{j-1}(d_0,1,\cdots,\hat{j},\cdots,2N)(d_1,j)]^2 \\
		&&=\sum_{{1\le i< j\le 2N}}(-1)^{i+j-1}(p_{i}^2-p_{j}^2)(d_0,1,\cdots,\hat{i},\cdots,2N)(d_0,1,\cdots,\hat{j},\cdots,2N)(i,j)\\
		&&-[\sum_{j=1}^{2N}(-1)^{j-1}(d_0,1,\cdots,\hat{j},\cdots,2N)(d_1,j)]^2\\
		&&=\sum_{{1\le i< j\le 2N}}(-1)^{i+j-1}(p_{i}^2-p_{j}^2)(d_0,1,\cdots,\hat{i},\cdots,2N)(d_0,1,\cdots,\hat{j},\cdots,2N)\delta_{2N+i-1,j}\\
		&&+\sum_{{1\le i< j\le 2N}}(-1)^{i+j-1}(p_{i}^2-p_{j}^2)(\dfrac{p_i-p_j}{p_i+p_j})(d_0,1,\cdots,\hat{i},\cdots,2N)(d_0,1,\cdots,\hat{j},\cdots,2N)(d_0,i)(d_0,j)\\
		&&-[\sum_{j=1}^{2N}(-1)^{j-1}(d_0,1,\cdots,\hat{j},\cdots,2N)(d_1,j)]^2
	\end{eqnarray*}
	Meanwhile, we can prove that the sum of the last two terms vanishes. 
	\begin{align*}
		&\sum_{{1\le i< j\le 2N}}(-1)^{i+j-1}(p_{i}-p_{j})^2 (d_0,1,\cdots,\hat{i},\cdots,2N)(d_0,1,\cdots,\hat{j},\cdots,2N)(d_0,i)(d_0,j)\\
		&-[\sum_{j=1}^{2N}(-1)^{j-1}(d_0,1,\cdots,\hat{j},\cdots,2N)(d_1,j)]^2\\
		&=\sum_{{1\le i< j\le 2N}}(-1)^{i+j-1}(p_{i}-p_{j})^2 (d_0,1,\cdots,\hat{i},\cdots,2N)(d_0,1,\cdots,\hat{j},\cdots,2N)(d_0,i)(d_0,j)\\
		&-\sum_{j=1}^{2N} p_{j}^2 (d_0,1,\cdots,\hat{j},\cdots,2N)^2 (d_0,j)^2\\
		&+2\sum_{{1\le i< j\le 2N}}(-1)^{i+j-1} p_{i}p_{j} (d_0,1,\cdots,\hat{i},\cdots,2N)(d_0,1,\cdots,\hat{j},\cdots,2N)(d_0,i)(d_0,j)\\
		&=\sum_{{1\le i< j\le 2N}}(-1)^{i+j-1}(p_{i}^2+p_{j}^2) (d_0,1,\cdots,\hat{i},\cdots,2N)(d_0,1,\cdots,\hat{j},\cdots,2N)(d_0,i)(d_0,j)\\
		&-\sum_{j=1}^{2N} p_{j}^2 (d_0,1,\cdots,\hat{j},\cdots,2N)^2 (d_0,j)^2\\
		&=\frac{1}{2} \sum_{i=1}^{2N}\sum_{j=1}^{2N} (-1)^{i+j-1}(d_0,1,\cdots,\hat{i},\cdots,2N)(d_0,1,\cdots,\hat{j},\cdots,2N)(d_0,i)(d_0,j)\\
		&=0
	\end{align*}
	Thus, the left hand of the equation (\ref{TODA}) can write as
	\begin{align}\nonumber
		&2((\partial_{x}^2 \tau_{k_1,k_2})\tau_{k_1,k_2}-(\partial_x \tau_{k_1,k_2})^2)\\\nonumber
		&=2\times \sum_{{1\le i< j\le 2N}}(-1)^{i+j-1}(p_{i}^2-p_{j}^2)(d_0,1,\cdots,\hat{i},\cdots,2N)(d_0,1,\cdots,\hat{j},\cdots,2N)\delta_{2N+i-1,j}\\\nonumber
		&=2\times \sum_{j=1}^{N} (p_{j}^2-{p^2_{2N+1-j}})(d_0,1,\cdots,\hat{j},\cdots,2N)(d_0,1,\cdots,\widehat{2N+j-1},\cdots,2N)
	\end{align}
	Therefore, if relation is written as 
	\begin{eqnarray*}
		\frac{2{a_1}^2 {\rho_1}^2} {(p^2_j-a^2_1)({p^2_{2N+1-j}}-a^2_1)} + \frac{2{a_2}^2 {\rho_2}^2}{(p^2_j-a^2_2)({p^2_{2N+1-j}}-a^2_2)} = -1,
	\end{eqnarray*}
	we obtain the bilinear equation (\ref{TODA}).
\end{proof}

\noindent {\bf Acknowledgements.}
Xiaochuan Liu's research was supported by National Science Foundation of China (Grant No. 12271424).

\end{document}